\documentclass[smallcondensed]{svjour3}

\smartqed

\usepackage{amsmath}
\usepackage{amsfonts}
\usepackage{amssymb}

\newcommand{\mket}[1]{| #1 \rangle}
\newcommand{\mbra}[1]{\langle #1 |}
\newcommand{\mbraket}[2]{\langle #1 | #2 \rangle}

\title{Generalised quantum weakest preconditions}

\author{Roman Gielerak \and Marek Sawerwain} 

\institute{
Roman Gielerak \at Institute of Control and Computation Engineering, \\
Faculty of Electrical Engineering, \\ Computer Science and Telecommunications, \\
University of Zielona G\'ora, Podg\'orna 50, 65-246 Zielona G\'ora, Poland.
\and 
Marek Sawerwain \at  Institute of Control and Computation Engineering, \\
Faculty of Electrical Engineering, \\ Computer Science and Telecommunications, \\
University of Zielona G\'ora, Podg\'orna 50, 65-246 Zielona G\'ora, Poland.
}

\date{Received: date / Accepted: date}

\begin{document}

\label{firstpage}
\maketitle

\begin{abstract}
\noindent Generalisation of the quantum weakest precondition result of D'Ho\-ndt and Panangaden is presented. In particular the most general notion of quantum predicate as positive operator valued measure (ter\-med POVM) is introduced. The previously known quantum weakest precondition result has been extended to cover the case of POVM playing the role of a quantum predicate. Additionally, our result is valid in infinite dimension case and also holds for a quantum programs defined as a positive but not necessary completely positive transformations of a quantum states.
\end{abstract}


\section{Introduction}

The formalism of 0--1 quantum predicates calculus was invented by von Neumann already in 1936 \cite{j_von_Neumann} and \cite{Mackey}. The main discovery was that the corresponding calculus significantly differs from the classical one (which is described in terms of the notion of Boolean algebra) and the development of the so--called quantum logic was achieved. In the past there were many activities in this fascinating area (for example see \cite{Piron} and \cite{DallaChiara_1977}, \cite{DallaChiara_2001}).

The recent developments in the quantum information area \cite{Nielsen_and_Chuang} renewed our interest in creating a general quantum predicate calculus in the context of the recent advances of quantum languages and quantum programming concepts \cite{Sanders2000}, \cite{Betelli}, \cite{Selinger2004} and \cite{Grattage06}. In the works \cite{Omer2003} and \cite{Omer2005} Bernard \"Omer introduced the first quantum programming language QCL. Paolo Zuliani also provides tools to compile quantum programs in \cite{Zuliani2005}. An extensive bibliographic review about the quantum programming theory and quantum languages was presented in \cite{Gay06}.

The weakest-precondition (in literature known as the weakest liberal precondition -- termed WP) is a well-known paradigm of a goal-directed programming methodology and semantics for a programming language. The weakest-precondition was developed in \cite{deBakker1} and \cite{deBakker2} and popularised in \cite{Dijkstra76}. This notion is connected with the Hoare triple $\{f_1\} P \{f_2\}$ \cite{Hoare69}, where $f_1$ and $f_2$ denote some predicates and $P$ is the program. In other words, the Hoare triple says: if $f_1$ is true for some entry state and after executing $P$ we obtain the final state, then $f_2$ is also true in the final state.

For any program $p$ and predicate $f_2$, we define the predicate $\mathrm{WP}(p, f_2)$ as
\begin{equation}
s \models \mathrm{WP}(p, f_2) \Leftrightarrow \forall_{t \in S} \;\; p(s) \longrightarrow t \Rightarrow (t \models f_2)
\label{weakest_precondition_operator}
\end{equation}
$\mathrm{WP}$ is the weakest precondition operator and the predicate $\mathrm{WP}(p, f_2)$ is the weakest one satisfying the Hoare triple $\{ \mathrm{WP}(p, f_2) \} p \{ f_2 \}$. The Hoare triple can be expressed with the $\mathrm{wp}$ operator $\models f_1 \Rightarrow \mathrm{WP}(p, f_2)$.

The strongest postcondition (termed SP) is defined by 
\begin{equation}
t \models \mathrm{SP}(p, f_1) \Leftrightarrow \exists_{s \in S} \;\; p(s) \longrightarrow t \; \wedge s \models f_1. 
\label{strongest_postcondition_operator}
\end{equation}
From the definition of the Hoare triple we obtain that $\models \mathrm{SP}(p, f_1) \Rightarrow f_2$ is equivalent to $\{f_1\} p \{f_2\}$.

In the work \cite{DHondt_Panangaden} the existence of weakest preconditions for quantum predicates defined as hermitian operators with spectral radius smaller than one was presented. However, the Kraus representation for completely positive finite dimensional superoperators was used in their proof in a very essential way. In particular the proof in \cite{DHondt_Panangaden} is valid when the following conditions are satisfied:
\begin{itemize}
\item[HP(1)] \hspace{0.15cm} the considered quantum systems are finite-di\-men\-sio\-nal,
\item[HP(2)] \hspace{0.15cm} the allowed quantum programs are defined as a completely positive transformations,
\item[HP(3)] \hspace{0.10cm} the admissible predicates are defined as hermitian operators with the operator norm smaller than one.
\end{itemize}
However, in many realistic situations all the listed assumptions HP(1)--HP(3) made in \cite{DHondt_Panangaden} are too restrictive. For example, a serious candidate for the realistic quantum computer, the computing machine with coherent pulses of light \cite{Ralph} is the point where infinite dimensional character of the corresponding quantum registers comes into play. Secondly in the so called active interpretation of predicate, the major role played by the very notion of predicates is to control the evolution of the state of quantum register\footnote{Another major role played by predicates is the role they play in the program developments as they describe state characteristic. This is so called passive (from the point of view of running program) interpretation which is important ingredient of the semantic analysis of computer programs}. There the well known problems connected with quantum measurements do occur. In particular the possible noisy character of quantum measurement is definitely excluded from the consideration by HP(3). In other words, the condition HP(3) restricts our considerations essentially to the orthodox von Neumann type of measurement only. Finally, although there are very plausible arguments in favour of completely positive maps as the only realistic transformations of the corresponding spaces of quantum states (the possible occurrence of positive maps in non-unitary quantum evolution is not definitely excluded \cite{Majewski}, \cite{Carteret_2008}).

In this paper we demonstrate result on the existence of the quantum weakest preconditions to cover the situations where none of the assumptions HP(1)--HP(3) are fulfilled. The main result is formulated precisely as Theorem (\ref{main_result}) in Sec.~(\ref{formulation_of_the_result:lbl}). What is surprising is that, the proof of our generalisation of the theorem due to \cite{DHondt_Panangaden} is very simple. The main argument is the use of Hilbert-Schmidt duality instead of Kraus representation as it was done in \cite{DHondt_Panangaden}.

\section{Formulation of the result}\label{formulation_of_the_result:lbl}

Let $\Sigma \subset \mathbb{R}^d$ be a Borel measurable subset of d-dimensional Euclidean space $\mathbb{R}^{d}$ and let $\mathcal{H}$ be a separable complex Hilbert space. $L(\mathcal{H})$ will stand for linear continuous operators on $\mathcal{H}$. The $\sigma$-algebra of sets of $\Sigma$ is denoted as $\beta(\Sigma)$.

A positive operator valued measure (POVM) on $(\Sigma, \beta(\Sigma))$ is a $\sigma$-additive map $\mathbb{F}$
\begin{equation}
\mathbb{F} : \beta(\Sigma) \longrightarrow L(\mathcal{H}) 
\end{equation}
and $F_{\Sigma} \leq \mathbb{I}_{\mathcal{H}}$, where $\mathbb{I}_{\mathcal{H}}$ is the unit operator on $\mathcal{H}$ . The space of such measurements will be denoted as $\mathrm{POVM}(\Sigma, \mathcal{H})$.

A natural partial order $\preceq$ can be defined in $\mathrm{POVM}(\Sigma, \mathcal{H})$. Let $\mathbb{F}, \mathbb{G} \in \mathrm{POVM}(\Sigma, \mathcal{H})$ then $\mathbb{F} \preceq \mathbb{G}$ iff
\begin{equation}
\forall_{A \in \beta(\Sigma)} \;\;\; \mathbb{F}(A) = {F}_A \preceq {G}_A = \mathbb{G}(A)
\end{equation}
where $\preceq$ is the natural ordering relation in $L(\mathcal{H})$ i.e.
\begin{equation}
F_A \preceq G_A \; \Leftrightarrow \; \forall_{\psi \in \mathcal{H}} \;\; \mbraket{\psi}{F_A\psi} \leq \mbraket{\psi}{G_A\psi}
\end{equation}

\begin{lemma}
For any Borel set $\Sigma \subset R^{d}$ the partially ordered space denoted as $(\mathrm{POVM}(\Sigma, \mathcal{H}), \preceq)$ is a completely partially ordered space (cpos).
\end{lemma}
\begin{proof}
Let $(\mathbb{F}^{(\alpha)})_{\alpha \in A}$ be any $\alpha$-ordered net in $\mathrm{POVM}(\Sigma, \mathcal{H})$. For any $\mket{\psi} \in \mathcal{H}$, $A \in \Sigma$ we define
\begin{equation}
\mbraket{\psi}{G_A\psi} = \sup_{\alpha} \mbraket{\psi}{F^{(\alpha)}_A\psi} .
\label{supremum_of_g}
\end{equation}
From the assumption on the uniform boundedness of elements from $\mathrm{POVM}(\Sigma, \mathcal{H})$ and the polarisation identities it follows that (\ref{supremum_of_g}) defines an operator $G_A$ and such that $|| G_{A} || \leq 1$. Let $(A_n)_{n=1,\ldots,\infty}$ be any family of pairwise disjoint subsets of $\beta(\Sigma)$, then the use of the version of the dominated convergence theorem allows us to formulate 
\begin{equation}
	G_{\bigcup_{n=1}^{\infty}A_n} = \Sigma^{\infty}_{n=1} G_{A_n}
\end{equation}
in the strong sense because of uniform boundedness in operator norm topology.

In conclusion, the operator valued map
\begin{equation}
A \in \beta(\Sigma) \longrightarrow G_A \in L(\mathcal{H})
\end{equation}
defines a POVM as shown above.
\end{proof}

By the (generalised)  complete quantum predicate we mean an arbitrary element $\mathbb{F} \in \mathrm{POVM}(\Sigma, \mathcal{H})$ which is complete. This means that $F_{\Sigma}=\mathbb{I}_{\mathcal{H}}$, where $\mathbb{I}_{\mathcal{H}}$ is the unit operator on $\mathcal{H}$. From the additivity of $\mathbb{F}$ it follows that for any $A \in \Sigma$, $\mathbb{F}_{A} \leq \mathbb{I}_H$ in the sense that:
\begin{equation}
\forall_{\mket{\psi} \in \mathcal{H}} \;\;\; \mbraket{\psi}{F_A \psi} \leq \mbraket{\psi}{\psi}
\end{equation}
from which it follows that for any $A \in \Sigma$ the operator norm $|| F_A ||$ of $F_A$ obeys $|| F_A || \leq 1$.

For $G \in L(\mathcal{H})$ let
\begin{equation}
r(G) = \sup \{ \, |\lambda| \; | \; \lambda \in \sigma(G) \, \}
\end{equation}
where $\sigma(G)$ means the spectrum of $G$. The number $r(G)$ is called the spectral norm of $G$. It is known that in the case of selfadjoint $G$ the spectral norm of $G$ is equal to the operator norm of $G$ \cite{Simon1970}.

Summarising our discussion let us note the following lemma:
\begin{lemma}
Let $\mathbb{F} \in \mathrm{POVM}(\Sigma, \mathcal{H})$ be a quantum predicate. Then for any $A \in \Sigma$:
\begin{equation}
|| F_A || = r(F_A) \;\;\; \mathrm{and} \;\;\; || F_A || \leq 1.
\end{equation}
\label{lbl:lemma:norm:predicate}
\end{lemma}
The set of all not necessary complete POVM on $(\Sigma, \mathcal{H})$ will be called the space of (generalised) quantum predicates and denoted as  $\mathrm{Pre}(\Sigma, \mathcal{H})$ and some times abbreviated as $\mathrm{Pre}(\mathcal{H})$ in the following.

\begin{remark}
In the paper \cite{DHondt_Panangaden} the assumption that the spectral norm of the hermitean operator (playing the role of quantum predicate there) is less or equal to one has been formulated. As we have remarked this is equivalent to the assumption that the operator norm is not exceding the value 1.
\end{remark}

For a given separable Hilbert space $\mathcal{H}$ the corresponding space of states $E(\mathcal{H})$ is usually defined as the set of non-negative, trace-class operators $\rho$ such that $\mathrm{Tr}(\rho) =  1$. A~space of admissible transformations of the space $E(\mathcal{H})$ is defined as the space of linear positive maps:
\begin{equation}
\mathbb{C}: \; E(\mathcal{H}) \longrightarrow E(\mathcal{H})
\end{equation}
that are trace preserving.

Any such map will be called (generalised) quantum program and the space of all such maps will be denoted as $\mathrm{QP}(\mathcal{H})$.

It is well known that the ring of trace-class operators on $\mathcal{H}$ denoted as $L_1(\mathcal{H})$ forms a two-sided $\star$-ideal in the $C^{\star}$-algebra $L(\mathcal{H})$ and therefore for any $\rho \in E(\mathcal{H})$, $\mathbb{F} \in \mathrm{Pre}(\Sigma, \mathcal{H})$ and $C \in \mathrm{QP}(\mathcal{H})$:
\begin{equation}
\mathrm{Tr}_{\mathcal{H}} ( \mathbb{F}_{A}C(\rho) ) \leq 1
\label{lbl:eq:trace:predicate}
\end{equation}

To prove (\ref{lbl:eq:trace:predicate}) we remark that any $A \in L(\mathcal{H})$ and $\rho \in L_1(\mathcal{H})$ the following inequality holds (Simon 1970, p.~218):
\begin{equation}
	{|| A\rho ||}_{1} \leq || A || \cdot {|| \rho ||}_1 .
\end{equation}
Thus, taking into account Lemma (\ref{lbl:lemma:norm:predicate}) the proof of (\ref{lbl:eq:trace:predicate}) follows.

The equation (\ref{lbl:eq:trace:predicate}) allows us to formulate the following lemma.

\begin{lemma}
For any $C \in \mathrm{QP}(\mathcal{H})$ the action $C^{\star}$ on $\mathrm{Pre}(\Sigma, \mathcal{H})$ is defined by:
\begin{equation}
\mathbb{F} \in \mathrm{Pre}(\Sigma, \mathcal{H}) \rightarrow C^{\star}\mathbb{F} \in \mathrm{Pre}(\Sigma, \mathcal{H})
\end{equation}
where
\begin{equation}
\forall_{\rho \in E(\mathcal{H})} \; \mathrm{Tr}((C^{\star}\mathbb{F}_{A})\rho) = \mathrm{Tr}(\mathbb{F}_A C(\rho))
\end{equation}
is action of $\mathrm{QP}(\mathcal{H})$ on the space $\mathrm{Pre}(\Sigma, \mathcal{H})$.
\end{lemma}
\begin{proof}
From $\rho \in L_1(\mathcal{H})$ and the spectral theorem it follows $\rho = \Sigma_{n} \lambda_n \mket{\psi_h}\mbra{\psi_n}$, $\lambda_n \geq 0$, $\lim_{n \to \infty} \lambda_n = 0$ and it is enough to assume that $\rho  = \mket{\psi}\mbra{\psi}$ for some $\mket{\psi} \in \mathcal{H}$ with  $| \mbraket{\psi}{\psi} | = 1$.

Therefore, the polarisation identities shows that the identities:
\begin{equation}
\mbraket{\psi}{(C^{\star}\mathbb{F}_{A})\psi} = \mathrm{Tr}(F_AC(\mket{\psi}\mbra{\psi}))
\end{equation}
define a bounded operator $C^{\star}\mathbb{F}$ for any $A \in \beta(\Sigma)$. The $\sigma$-additivity of $C^{\star}\mathbb{F}$ is also easy to prove.
\end{proof}

The duality between $L(\mathcal{H})$ and $L_1(\mathcal{H})$ based on $\mathrm{Tr}$ will be called Hilbert-Schmidt duality and will be also denoted as
\begin{equation}
{\mbraket{ \cdot }{ \cdot }}_{HS} : (\rho, A) \in L_1(\mathcal{H}) \times L(\mathcal{H})  \longrightarrow {\mbraket{\rho}{A}}_{HS} = \mathrm{Tr}(\rho A)
\end{equation}

\begin{definition}
For a given $\mathcal{H}$, $\mathbb{F} \in \mathrm{Pre}(\Sigma, \mathcal{H})$, $\rho \in E(\mathcal{H})$ the function 
\begin{equation}
\mathrm{sat}(\rho, \mathbb{F}) : A \in \Sigma \longrightarrow \mathrm{sat}(\rho, F_A) = \mathrm{Tr}(\rho F_A)
\end{equation}
will be called the satisfiability of the quantum predicate $\mathbb{F}$ in the state $\rho$. In particular the state $\rho$ satisfies the predicate $\mathbb{F}$ iff the function $\mathrm{sat}(\rho, F)$ is nonzero positive valued.
\label{sat_def}
\end{definition}

From Def.~(\ref{sat_def}) it follows that the function of satisfiability $\mathrm{sat}(\rho, \mathbb{F})$ is always a bounded measure on $\beta(\Sigma)$.

\begin{definition}
An $s$-order, denoted as $\stackrel{s}{\preceq}$ is defined on the space of quantum predicates $\mathrm{Pre}(\Sigma, \mathcal{H})$ in the following way:
\begin{equation}
\mathbb{F} \stackrel{s}{\preceq} \mathbb{G} \;  \mathrm{iff} \; \forall_{A \in \beta(\Sigma)} \; \forall_{\rho \in E(\mathcal{H})} \; \mathrm{sat}(\rho, \mathbb{F})(A) \leq \mathrm{sat}(\rho, \mathbb{G})(A)
\end{equation}
\end{definition}

Similarly, it can be proved that the semi-ordered space $(\mathrm{Pre}(\Sigma, \mathcal{H}), \stackrel{s}{\preceq})$ is a cpos.

\begin{lemma}
For $\Sigma \subset \mathbb{R}^{d}$ and separable Hilbert space $\mathcal{H}$ the semi-ordered space denoted as $(\mathrm{Pre}(\Sigma, \mathcal{H}), \stackrel{s}{\preceq})$ is completely partially ordered space.
\end{lemma}

For a given quantum predicate $\mathbb{F} \in \mathrm{Pre}(\Sigma, \mathcal{H})$ the set of preconditions for $\mathbb{F}$ with respect to quantum program $C \in \mathrm{QP}(\mathcal{H})$ and denoted as $\{C\}(\mathbb{F})$ is defined as:
\begin{equation}
 \{ C \}(\mathbb{F}) = \{ \mathbb{G} \in \mathrm{Pre}(\Sigma, \mathcal{H}) \} : \mathbb{G}  \stackrel{s}{\preceq} \mathbb{F} \}
\end{equation}

\begin{definition}
A weakest precondition for a predicate $\mathbb{F} \in \mathrm{Pre}(\Sigma, \mathcal{H})$ with respect to a quantum program $C \in \mathrm{QP}(\mathcal{H})$ denoted (if it exists) as $\mathrm{WP}(C)(\mathbb{F})$ is the predicate $\mathbb{G} \in \mathrm{Pre}(\Sigma, \mathcal{H})$ such that $\mathbb{G} = \mathrm{WP}(C)(\mathbb{F}) = \sup( \{C\}\mathbb{F})$.
\end{definition}

\begin{theorem}
Let $\mathcal{H}$ be a separable Hilbert space and let $C \in \mathrm{QP}(\mathcal{H})$ be a given quantum program and let $\mathbb{F} \in  \mathrm{Pre}(\Sigma, \mathcal{H})$. Then there exists unique $\mathbb{G} \in \mathrm{Pre}(\Sigma, \mathcal{H})$ such that
\begin{equation}
	\mathbb{G} = \mathrm{WP}(C)\mathbb{F}
\end{equation}
\label{main_result}
\end{theorem}
\begin{proof}
By taking $A \in \beta(\Sigma)$, we can write
\begin{equation}
\mathrm{Tr}(F(A)C(\rho)) = {\mbraket{F_A}{C\rho}}_{HS(\mathcal{H})} = {\mbraket{C^{\star}F_A}{\rho}}_{HS(\mathcal{H})} .
\end{equation}
By the Hilbert-Schmidt duality we can define a new $\mathrm{POVM}(C^{\star}\mathbb{F})$ on $(\Sigma, \mathcal{H})$ by the last identity. Thus, we can expect that $C^{\star}\mathbb{F} = \mathrm{WP}(C)\mathbb{F}$. Let $\mathbb{H} \in \{C\}(\mathbb{F})$, then for some $A \in \beta(\Sigma)$
\begin{equation}
\mathrm{Tr}(\mathbb{H}_{A}\rho) \leq \mathrm{Tr}(F_AC(\rho)) = {\mbraket{F_A}{C\rho}}_{HS} = {\mbraket{C^{\star}F_A}{\rho}}_{HS} = \mathrm{Tr}(C^{\star}\mathbb{F}_A)\rho
\end{equation}
and
\begin{equation}
\mathrm{sat}(\mathbb{H}, \rho) \leq \mathrm{sat}(C^{\star}\mathbb{F}, \rho)
\end{equation}
then $C^{\star}\mathbb{F}$ is majorising for the set $\{C\}\mathbb{F}$. Obviously, from the very construction of $C^{\star}\mathbb{F}$ it follows that $C^{\star}\mathbb{F} \in \{C\}\mathbb{F}$.
\end{proof}

\begin{remark}
The action of $\mathrm{CP}(\mathcal{H})$ on the spaces $\mathrm{POVM}(\mathcal{H})$ were studied more carefully in \cite{Buscemi} and some very interesting results on this were obtained. Whether those results can be extended to the action of $\mathrm{QP}(\mathcal{H})$ and whether this kind of results could be efficient in the quantum programming area  in our opinion deserve further studies.
\end{remark}

\begin{remark}
The theorem presented in work \cite{DHondt_Panangaden} is a special case of our theorem \ref{main_result}. If we assume that the quantum predicate is given by the corresponding $\mathbb{F} \in \mathrm{POVM}(\Sigma, \mathcal{H})$ with one atom support,
\begin{equation}
\mathbb{F}= \{ F_1 \}
\end{equation}
then our theorem gives (still with some generalisation) the D'Hondt and Panangaden result.
\end{remark}

\begin{remark}
In the case of quantum programms defined as completely positive maps the infinite dimensional version  (the $C^{\star}$ version) of the Kraus theorem known as Stinespring representation theorem (saying that any unital CP map on $C^{\star}$-algebra is the compression of some inner $\star$-homorphism) can be used instead of the use of Hilbert-Schmidt duality \cite{Paulsen2003}. However the corresponding constructions are much more complicated as we have to pass to the corresponding dilation spaces.
\end{remark}

\section{Summary and conclusions}

The most general notion of quantum predicate using the notion (connected to an a priori noisy measurement process) of positive operator valued measures has been introduced in this note.

The existence of the corresponding quantum weakest preconditions has been proved. Additionally, our result is valid in infinite dimensional situations and for positive but not necessary completely positive quantum programs.

It would be of great importance to provide some examples showing that our generalised quantum predicate notion can be used for semantic analysis of quantum programs. Especially important seems to be the question of providing interesting examples where previously known tools and results are not directly applicable. This will be a main topic of a forthcoming paper \cite{Gielerak_and_Sawerwain}.


\begin{thebibliography}{}

\bibitem{deBakker1}
de Bakker J. W., de Roever W. P.: {\it A calculus for recursive programs schemes}, In: Automata, Languages, and Programming, Amsterdam, North-Holland, pp.~167--196, 1972.

\bibitem{deBakker2}
de Bakker J. W., Meertens, L. G. L. T.: {\it On the completeness of the inductive assertion method}, J. Comput. Syst. Sci. Vol.~11, No.~3, pp.~323-357, 1975.

\bibitem{Betelli} Betelli S., Serafini L. and Calarco T.: {\it Toward an architecture for quantum programming}, Eur. Phys. J., 25:181--200, 2003, arXiv:cs/0103009.

\bibitem{j_von_Neumann} Birkhoff G., von Neumann J.: {\it The Logic of Quantum Mechanics}, Ann. Math., Vol.~37, pp.~823-843, 1936.

\bibitem{Buscemi} Buscemi F., D'Ariano G.M., Keyl M., Perinotti P., Werner R.F.: {\it Clean positive operator valued measures}, Journal of Mathematical Physics, Vol.~46, Issue~8, p.082109, 2005, arXiv:quant-ph/0505095v5.

\bibitem{Carteret_2008} Carteret H.A., Terno D.R., \.Zyczkowski K.: {\it Dynamics beyond completely positive maps: Some properties and applications}, Phys. Rev. A 77, 042113, 2008, arXiv:quant-ph/0512167v3.

\bibitem{DallaChiara_2001} Dalla Chiara, M.L.: {\it Quantum logic}, in: D.M. Gabbay, F. Guenthner (Eds.), Handbook of Philosophical Logic, Vol. III, 1986, pp. 427--469. Revised version in: Handbook of Philosophical Logic, Vol. 6, 2nd edn., Kluwer, Dordrecht, pp.~129--228, 2001.

\bibitem{DallaChiara_1977} Dalla Chiara M.L.: {\it Quantum logic and physical modalities}, J. Philos. Logic 6, pp.~391--404, 1977.

\bibitem{Choi} Choi M.D.: {\it Completely positive linear maps on complex matrices}, Linear Algebra and its Applications, Vol.~10, pp.~285--290, 1975.

\bibitem{Dijkstra76} Dijkstra E. W.: {\it A Discipline of Programming}, Prentice-Hall, Englewood Cliffs, N.J., 1976.

\bibitem{Gay06} Gay S.J.:  {\it Quantum Programming Languages: Survey and Bibliography}, Mathematical Structures in Computer Science Vol.~16, No.~4, 2006.

\bibitem{Grattage06} Grattage J.: {\it QML: A functional quantum programming language}, PhdThesis, 2006.

\bibitem{Gielerak_and_Sawerwain} Gielerak R., Sawerwain M.: {\it General quantum predicate as semantics tools in quantum programming theory}, in preparation.

\bibitem{Hoare69} Hoare C.: {\it An axiomatic basis for computer programming}, Communications of the ACM, Vol.~12, pp.~576--583, 1969.

\bibitem{DHondt_Panangaden} D'Hondt, E., Panangaden, P.: {\it Quantum weakest preconditions}, Mathematical Structures in Computer Science, Vol.~16, No.~3, pp.~429--451, 2006.

\bibitem{Kraus83} Kraus K.:  {\it State, Effects, and Operations}, Berlin, Springer-Verlag, 1983.

\bibitem{Mackey} Mackey G.: {\it Mathematical foundations of Quantum Mechanics}, W.A. Benjamin, 1963.

\bibitem{Majewski} Majewski W.A.: {\it On non-completely positive quantum dynamical maps on spin chains}, Phys. A: Math. Theor. 40, pp.~11539-11545, arXiv:quant-ph/0606176x2, 2007.

\bibitem{Nielsen_and_Chuang} Nielsen M., Chuang I. L.: {\it Quantum Computation and Quantum Information}, Cambridge University Press, 2000.

\bibitem{Omer2003} \"{O}mer B.: {\it Structured Quantum Programming}, PhD thesis, Technical University of Vienna, Austria, 2003. 

\bibitem{Omer2005} \"{O}mer B.: {\it Classical Concepts in Quantum Programming}, Int. J. of The. Phys., Vol.~44, No.~7,pp.~943--955, 2005.

\bibitem{Piron} Piron P.: {\it Foundations of Quantum Physics}, W. A. Benjamin, 1976.

\bibitem{Paulsen2003} Paulsen V.: {\it Completely Bounded Maps and Operator Algebra}, Cambridge University Press, 2003.

\bibitem{Ralph} Ralph T.~C.,  Gilchrist A., Milburn  G.~J., Munro W.~J., Glancy S.: {\it Quantum computation with optical coherent states}, Phys. Rev. A, Vol.~68, 042319, 2003.

\bibitem{Sanders2000} Sanders J. W., Zuliani P.: {\it Quantum programming}, In Mathematics of Program Construction, Springer LNCS 1837, pp.~80--99, 2000.

\bibitem{Selinger2004} Selinger P.: {\it Towards a quantum programming language}, Mathematical Structures in Computer Science, Vol.~14, Issue 4, pp.~527-586, 2004.

\bibitem{Simon1970} Simon.B: {\it Methods of Modern Mathematical Physics.}, Academic Press, New York and London, 1970.

\bibitem{Zuliani2005} Zuliani P.: {\it Compiling quantum programs}, Acta Informatica, Vol. 41, Issue 7-8, pp.~435--474, 2005.


\end{thebibliography}
\end{document}